\documentclass[10pt]{article}
\usepackage[latin1]{inputenc}
\usepackage{amsmath,amssymb,amsthm}
\usepackage{graphicx}
\usepackage{epsfig}

\newtheorem{theorem}{Theorem}[section]

\newtheorem{definition}[theorem]{Definition}

\newtheorem{lemma}[theorem]{Lemma}

\newenvironment{remark}[1][Remark]{\textbf{#1:} }{ ~}

\newcounter{listagem}
\newcommand{\blista}{\begin{list}{\roman{listagem})}{\usecounter{listagem}}}
\newcommand{\elista}{\end{list}}

\newcommand{\beq}{\begin{equation}}
\newcommand{\eeq}{\end{equation}}
\newcommand{\beqn}{\begin{eqnarray}}
\newcommand{\eeqn}{\end{eqnarray}}
\newcommand{\ov}{\overline}
\newcommand{\ul}{\underline}
\def\le{<}
\def\ge{>}
\newcommand{\Cl}{{C \kern -0.1em \ell}}

\newcommand{\BR}{\mathbb{R}}
\newcommand{\BC}{\mathbb{C}}

\newcommand{\K}{\mathbf{k}}

\newcommand{\ds}{\displaystyle}

\begin{document}

\title{The Schr\"odinger semigroup on some flat and non flat manifolds}
\author{R. S. Krau{\ss}har$^\dag$  ~~~~ M.M. Rodrigues$^\ddag$ ~~~~ N. Vieira$^*$\\ \\
{\small $^\dag$ Fachbereich Mathematik} \\ {\small Technische Universit\"at Darmstadt }\\{\small Schlo{\ss}gartenstr. 7}\\{\small 64289 Darmstadt, Germany.} \\ {\small E-mail: krausshar@mathematik.tu-darmstadt.de} \\ \\
{\small $^\ddag$ Department of Mathematics} \\ {\small University of Aveiro} \\ {\small Campus Universit\'ario de Santiago} \\ {\small 3810-193 Aveiro, Portugal} \\ {\small E-mail: mrodrigues@ua.pt} \\ \\
{\small $^*$ Center of mathematics of University of Porto} \\ {\small Faculty of Science} \\ {\small University of Porto} \\ {\small Rua do Campo Alegre} \\ {\small 4169-007 Porto, Portugal} \\ {\small E-mail: nvieira@fc.up.pt}} \maketitle


\begin{abstract}
In this paper we apply known techniques from semigroup theory to the Schr\"odinger problem with initial conditions. To this end, we  define the regularized Schr\"odinger semigroup acting on a space-time domain and show that it is strongly continuous and contractive in $L_p,$ with $\frac{3}{2} < p < 3.$  These results can easily be extended to the case of conformal operators acting in the context of differential forms, but they require positiveness conditions on the curvature of the considered Minkowski manifold. For that purpose, we will use a Clifford algebra setting in order to highlight the geometric characteristics of the manifold. We give an application of such methods to the regularized Schr\"odinger problem with initial condition and we will extended our conclusions to the limit case. For the torus case and a class of non-oriented higher dimensional M\"obius strip like domains we also give some explicit formulas for the fundamental solution.
\end{abstract}



\section{Introduction}

One of the most important PDE's is the Schr\"odinger equation. Physically, this equation describes the space and time dependence of quantum mechanical systems. It is of extreme importance to the theory of quantum mechanics, playing a role analogous to Newton's second law in classical mechanics. In the mathematical formulation of quantum mechanics, each system is associated with a complex Hilbert space such that each instantaneous state of the system is described by a unit vector in that space. This state vector encodes the probabilities for the outcomes of all possible measurements applied to the system. As the state of a system generally changes over time, the state vector is a function of time. The Schr\"odinger equation provides a quantitative description of the rate of change of the state vector.

Formally, the Schr\"odinger equation is expressed by
\beqn
H(x) \psi(x,t) = \pm i \hbar \partial_t \psi(x,t), \nonumber
\eeqn
where $i$ is the imaginary unit, $x$ the space-variable, $t$ the time-variable, $\partial_t$ is the partial derivative with respect to $t$, $\hbar$ is the reduced Planck's constant (Planck's constant divided by $2\pi$), $\psi(x,t)$ is the wave function, $H(x)$ is the Hamiltonian (self-adjoint operator acting on the space variable), and $\pm$ represents the forward or backward case, respectively.

The Hamiltonian describes the total energy of the system. In analogy to the occurrence of the force in Newton's second law, its exact form is not provided by the Schr\"odinger equation. It must be independently determined by physical properties of the system.

In order to simplify the calculations in this paper we omit the reduced Planck's constant, and we will concentrate ourselves on the backward case. Nevertheless, all the theoretical results that we present can directly be adapted to the forward case (for more details about the Schr\"odinger equation, see for instance \cite{BS} and \cite{Sch}).

There are several areas of Mathematics that can be applied in the study of PDE's. However, most of them are only  efficient when we  deal with elliptic operators and fail in the context of parabolic and hyperbolic operators, as for example, in the case of the Schr\"odinger operator or the heat operator. In this paper, we will try to apply some of the elliptic techniques used to study the heat problem in the analysis of the Schr\"odinger problem. Nevertheless, we need to take into account that in many aspects the Schr\"odinger operator is substantially different from the heat operator. For example, notice that the Galilean group is the invariance group associated to the first equation, while the parabolic group is the invariance group associated to the heat equation (see \cite{T}). The Schr\"odinger equation is related to the Minkowski space-time metric, while the heat equation is linked to the parabolic space-time metric (see \cite{T}).  Also, and more important for our consideration, under an analytical point of view, the singularity $t=0$ of the corresponding  fundamental solutions is removable outside the origin in the second case. However, in the case dealing with the Schr\"odinger equation, this is not true. This fact forces us to introduce a regularization procedure prior to the treatment by semigroup theory or hypoelliptic theory (see \cite{CV}, \cite{KV1}, \cite{T} and \cite{V}).

In this paper we consider an approach that combines semigroup theory with Clifford analysis methods and provides a successful solution of the Schr\"odinger problem. The main results presented here are based on Eichhorn's ideas (see \cite{E}). In his paper the author presents the heat semigroup acting either on tensors or differential forms, with values in a vector bundle and applies it to solve the heat problem with initial data. The implementation of a regularization procedure allows an extension of the results related with the heat operator to the regularized Schr\"odinger operator.

Hence, the paper is structured as follows:  in Section 2 we present the necessary notions regarding Clifford analysis, G\"unter derivatives, semigroup theory, differential forms and we describe our regularization procedure. In the following section, we use the ideas in \cite{E} and \cite{G} to construct the regularized Schr\"odinger semigroup and to prove under which conditions the Laplacian is dissipative in $L_p$, independently of considering flat or non-flat domains.
We give some explicit representation formulas for the solutions of the Schr\"odinger problem on $n$-tori associated with different spin structures. Then we explain how to adapt these constructions to a class of non-orientable flat manifolds that consists of higher dimensional generalizations of the M\"obius strip and the Klein bottle.
\par\medskip\par
To study the case of non-flat manifolds we will consider the Bochner and G\"unter-Laplacians acting on differential forms. In Section 4 we will use some of the properties of the obtained semigroup to solve the regularized Schr\"odinger problem with initial condition. In the last section we finally explain how we can extend the results presented in Section 4 to a geometrically more general framework.


\section{Preliminaries}

\subsection{Clifford analysis}\label{sec2.1}

We consider the $n-$dimensional vector space $\BR^n$ endowed with an orthonormal basis $\{e_1,\ldots,e_n\}.$ We define the universal Clifford algebra $\Cl_{0,n}$ as the $2^n-$di\-men\-sional associative algebra which preserves the multiplication rules $e_ie_j + e_j e_i = -2\delta_{i,j}.$ A basis for $\Cl_{0,n}$ is given by $e_0=1$ and $e_A = e_{h_1}\ldots e_{h_k},$ where $A=\{h_1,\ldots,h_k\} \subset \{1,\ldots,n\},$ for $1 \leq h_1 < \ldots < h_k \leq n.$ Each element $x \in \Cl_{0,n}$ will be represented by $x = \sum_A x_A e_A,$ $x_A \in \BR.$ Let $\BC_n = \Cl_{0,n} \otimes \BC$ the complexification of the universal Clifford algebra presented previously.

We introduce the Euclidean Dirac operator $D = \sum_{j=1}^n e_j \partial_{x_j}$ associated to the flat metric $ds^2=dx_1^2+\cdots+ dx_n^2$. It factorizes the $n-$di\-men\-sional Laplacian, that is, $D^2 = - \Delta.$ A $\BC_n-$valued function defined on an open domain $\Omega,$ $u: \ul{\Omega} \subset \BR^n \rightarrow \BC_n$ is said to be left monogenic if it satisfies $Du = 0$ on $\Omega$ (resp. right-monogenic if it satisfies $uD=0$ on $\Omega$). We remark that whenever $u$ is scalar, $Du$ coincides with the gradient $\nabla u$. For more details about monogenic functions, we refer the reader for instance to \cite{DSS} or elsewhere.

We further say that a $\BC_n$-valued function $u$ belongs to a certain function space if and only if all its coordinate-functions $u_A$ belong to the corresponding (real or complex) function space. For instance, $ u = \sum_A u_A e_A$ belongs to $L_p (\Omega, \BC_n)$ if and only if all its complex valued coordinate functions $u_A$ are in $L_p(\Omega, \BC_n).$ Whenever no confusion arises, the $\BC_n-$valued function spaces will be denoted by the same notation of its real counterparts, that is, $L_p(\Omega, \BC_n)$ will be identified with $L_p(\Omega)$. For general (Clifford algebra valued) $L_p$ spaces, the usual $L_p$-norm is defined by
\beqn
||u||_{L_p} := \left ( \int_\Omega |u(x)|^p ~ dx \right)^{\frac{1}{p}} < \infty, \nonumber
\eeqn
with $1 \leq p < \infty$ and
\beqn
|u|^2 := 2^n \left[u ~ \ov{u} \right]_0 = 2^n \sum_A |u_A|^2 e_A \ov{e}_A = 2^n \sum_A |u_A|^2. \nonumber
\eeqn
Here, $[\cdot]_0$ denotes the scalar part of the element.

For $p \not= 2$ they are Banach spaces while $L_2(\Omega)$ can be extended to a Hilbert space by introducing an inner product of the form
\beqn
<u, v> := \int_{\Omega} u(x) | v(x) dx = 2^n \int_{\Omega} \left[u(x) ~ \ov{v}(x) \right]_0 dx. \nonumber
\eeqn
Here, $u, v \in L_2(\Omega)$ and
\beqn
u(x) ~ \ov v(x) = u ~ \ov{v} := \sum_{A, B \subset N} u_A ~ \ov{v}_B ~ e_A ~ \ov{e}_B. \nonumber
\eeqn


\subsection{Regularization of the non-stationary Schr\"odinger operator}\label{subsec2.2}

The following fundamental solution of the time-dependent Schr\"odinger operator
\beqn
e_-(x,t) = i ~ \frac{H(t)}{(4 \pi i t)^\frac{n}{2}} \exp \left(-i ~{\frac{|x|^2}{4t}}  \right). \nonumber
\eeqn
has non-removable singularities in the whole hyperplane $t=0$. This is one reason why one cannot directly apply the methods of hypoelliptic operators. This feature carries additional problems for the study of the arising integral operators, where we cannot guarantee the convergence (in the classical sense) of the integrals that define those operators.

In order to solve this problem we need to regularize the fundamental solution and the arising operators (see \cite{CV}, \cite{T} and \cite{V}). This process of regularization creates a family of operators and corresponding fundamental solutions, which are locally integrable over $\ds \BR^n \times \BR_0^+ \setminus \{ \ul{0}, 0 \}$. Moreover, this family will converge to the original operators and fundamental solutions when we consider the limit process $\epsilon \rightarrow 0^+$.

To this end, we will replace the imaginary unit appearing in the Schr\"odinger operator by the value $\ds \K=\frac{\epsilon + i}{\epsilon^2 + 1}$ and we obtain a new operator $-\Delta \pm \K \partial_t$. For each $\epsilon > 0$ the associated operator $-\Delta \pm \K \partial_t$ is a hypoelliptic operator, in the sense of Theorem 1.8 presented in Section 1.3 of \cite{ABN}. This modification has a good behavior of the associated integral operators. More details about the regularization of the Schr\"odinger operator can be found in \cite{CV}, \cite{KV1}.


\subsection{Basic notions in semigroup theory}

Consider an operator $F: D_F \subset X \rightarrow X$ where we assume that $D_F$ is a dense set in $X$ and that $F$ is a closed operator. First we introduce the following characterization of a normalized tangent functional via the complex version of the Hahn-Banach theorem.
\begin{theorem} (c.f. \cite{HHZ})
Let $X$ be a complex Banach space and $Y$ be a linear subspace of $X$. If $u \in Y^\ast$, then there exist a normalized tangent functional $u^\ast \in X^\ast$ such that $u^\ast|_Y = u$ and $||u^\ast||_{X^\ast} = ||u||_{Y^\ast}$.
\end{theorem}

Taking into account the previous result, $F$ is called dissipative if for every $u \in D_F$ there exists a normalized tangent functional such that $\langle u^\ast, Fu \rangle \leq 0$. The closure of a dissipative operator is dissipative. For the particular case where $X$ is  a Hilbert space and $F$ a symmetric operator, the condition $\langle Fu,u \rangle \leq 0$ for all $u \in D_F$ implies that $F$ is dissipative. We say that a $C^0-$semigroup $\{ T_t\}_{t \in \BR_0^+}$ of bounded linear operators $T_t \in L(X,X)$, where $X$ is a Banach space, is called a contraction semigroup if $||T_t|| \leq 1$, for $0 \leq t < +\infty.$ The infinitesimal generator $F$ of a contraction semigroup can be characterized by the following property.
\begin{lemma}\label{lem1}
(c.f. \cite{RS})
Suppose that $D_F$ is dense. A closed operator $F:D_F \rightarrow X$ is the infinitesimal generator of a contraction semigroup if and only if $F$ is dissipative and $\mathrm{Range}(\mu - F) = X$, for some $\mu > 0.$
\end{lemma}


\subsubsection{The Minkowski metric}

A pseudo-Riemannian metric on a smooth manifold $M$ is a symmetric 2-tensor field $g$ that is non-degenerate at each point $x \in M$. By far the most important pseudo-Riemannian metrics are the Lorentz metrics, which are pseudo-Riemannian metrics of index 1. The standard example of a Lorentz metric is the Minkowski metric, that is, a metric $g$ on $\BR^{n+1}$ that is written in terms of the local coordinates $(\xi_1,\ldots,\xi_n,\tau)$ as
\beqn
g(d\overrightarrow{\xi},d\overrightarrow{\xi})=(d\xi_1)^2 + \ldots + (d\xi_n)^2 - (d\tau)^2. \label{3.5}
\eeqn

The separation or difference of the physical characteristics of the space coordinates (the $\xi$ directions) and the time coordinate (the $\tau$ direction) arises from the fact that they are subspaces on which $g$ is positive or negative definite, respectively.


\subsection{Differential forms theory}

Here, we recall some basic definitions from the theory of differential forms.

\begin{definition}
The space $\bigwedge_k \Omega$ of differential $k-$ forms at $x$ is the set of all $k-$linear alternating functions
 $$\omega : \Omega \times \cdots \Omega \rightarrow \Omega.$$
\end{definition}

The space $\bigwedge_k \Omega $ is a vector space under the operations of addition and scalar multiplication.

Let $L_p(\bigwedge_k \Omega)$ be the corresponding space of  $k-$forms with values in $\BC_n$ and $L_p^0(\bigwedge_k \Omega)$ of those which have a compact support.

Following \cite{L}, we now present the Laplace operator in the context of differential forms. The concept of harmonic functions can be extended to differential forms as follows. Let $\star$ denotes the Hodge star operator. The latter is a linear operator acting as
\begin{align*}
& \star (1) = \pm dx_1 \wedge dx_2 \wedge \ldots \wedge dx_n, \\
& \star (dx_1 \wedge dx_2 \wedge \ldots \wedge dx_n) = \pm 1, \\
& \star (dx_1 \wedge dx_2 \wedge \ldots \wedge dx_p) = \pm dx_{p+1} \wedge \ldots \wedge dx_n,
\end{align*} where the $\pm$ sign corresponds to the positive or negative orientation, respectively, of the form $dx_1 \wedge dx_2 \wedge \ldots \wedge dx_p.$

We introduce its adjoint $d^\ast$ acting on $k-$forms by setting $d^\ast=(-1)^{n(p+1)+1} \star d \star.$ While the exterior differentiation operator maps $k-$forms to $(k+1)-$forms, its adjoint maps $k-$forms to $(k-1)-$forms. A $k-$form $\omega$ is called  harmonic if and only if it is closed ($d\omega =0$) and co-closed ($d^\ast \omega =0$). Then we introduce the Hodge Laplacian, also called Laplace-Beltrami operator, by $\Delta_H = d^\ast d + dd^\ast.$

Moreover, differential forms are also used to express tensorial actions. However, in view of \cite{CBCF} and \cite{Graf}, we can identify tensors on $\Omega$ with elements of the universal Clifford algebra $\Cl_{0,n}$. This fact allows us to avoid the use of vector bundles, metric connections, and other heavy machinery used in \cite{E}.


\subsection{Laplace operators on manifolds}

Until now we only have considered domains in $\BR^n$ with the standard Euclidean Laplace operator $\Delta = \sum_{i=1}^n \partial_{x_i}^2.$ We now look into more complex structures of manifolds endowed with an arbitrary metric. In this section, we  introduce the Bochner-Laplacian and G\"unter-Laplacian, the operators which reflect the new metric structure.

The Bochner-Laplacian is given by $\Delta_B = \nabla^\ast \nabla$, where $\nabla^\ast$ stands for the formal adjoint of the L\'evi-Civita connection (for more details see \cite{DMM}). This Laplacian and the Euclidean one introduced in Subsection \ref{sec2.1} are related by the following special case of the Weitzenbock identity, proved in \cite{GM},
\begin{eqnarray}
\nabla^\ast \nabla = -\Delta - \mathrm{Ric},\label{c1.5*}
\end{eqnarray}
where $\mathrm{Ric}$ is the Ricci curvature on $\Omega$.

It is known that one possible extension of the most basic partial differential operators on an domain  $\ul{\Omega} \subset \BR^n$, can be expressed globally, in terms of the standard spatial coordinates in $\BR^n$. It turns out that a  convenient way to carry out this program is to employ the so-called G\"unter derivatives  (for more details see \cite{DMM} and \cite{Gun})
\beqn
\mathcal{D}:=(\mathcal{D}_1, \mathcal{D}_2, ..., \mathcal{D}_n), \label{c1.3*}
\eeqn
where for each $1 \leq j \leq n$, the first-order differential operator $\mathcal{D}_j$ is the directional derivative along $\psi e_j$, where $\psi: \BR^n \rightarrow T_x\Omega$ is the orthogonal projection onto the tangent plane to $\Omega$ and, as usual, $e_j = (\delta_{j,k})_{1 \leq k \leq n} \in \BR^n$, with $\delta_{jk}$ denoting the Kronecker symbol. The operator $\mathcal{D}$ is globally defined on $\Omega$ by means of the unit normal vector field, and has a relatively simple structure. In terms of (\ref{c1.3*}), the Laplace operator defined via G\"unter derivatives, namely the G\"unter-Laplacian,  becomes
\beqn
\Delta_G = \mathcal{D}^2 = \sum_{j=1}^n \mathcal{D}_j^2 = \sum_{j=1}^n (\partial_{x_j} - \nu \partial_{\nu})(\partial_{x_j} - \nu \partial_\nu), \nonumber
\eeqn
with $\nu(x) := \frac{x}{||x||},$ $x  \in \BR^n \setminus \{ 0 \}$, and where $\partial_\nu = \sum_{j=1}^n \left( \frac{x_j}{||x||} \right) \partial_{x_j}$ is the radial derivative in $\BR^n$. For the Laplace operator introduced in Subsection \ref{sec2.1} and $\Delta_G$. We have the following identity
\beqn
\Delta = \psi \mathcal{D}^2 + 2R^2 - \mathcal{G}R, \label{c1.4*}
\eeqn
where $R(x)= \nabla \nu (x)$ and $\mathcal{G} = \mathrm{div} \nu$. Relations (\ref{c1.3*}) and (\ref{c1.4*}) are proved in \cite{DMM}.


\section{The regularized Schr\"odinger semigroup acting on vector bundles}

In this section the main objective is to construct the regularized semigroup associated to our operator, namely $\{ \Gamma_t^\K \}_{t \in \BR_0^+}$,  and to show that, under specific values of $p$, we can use it to obtain a unique solution of the regularized  Schr\"odinger equation in $L_p$. The application to the solution of the equation will only be possible after we study the dissipativity of the elements of the semigroup.


\subsection{Semigroups associated to regularized Schr\"odinger operators}

The use of the semigroups techniques in the study of time-evolution equations has several advantages. For example, they provide an elegant alternative to establish existence results for evolution equations. Important connections between semigroup theory and the Schr\"odinger equation have already been established by a number of authors. In \cite{Zhao} for instance, the author constructed the associated semigroup via the infinitesimal generator without using any type of regularization procedure or the spectral theorem.

In this section we want to construct the semigroup associated to our evolution operator in a simplest possible way. This construction is based on some ideas presented in \cite{E} by Eichhorn. The main difference between his and our approach is that we cannot use the Schr\"odinger operator itself. This impossibility is due to the fact that our time-dependent operator is not hypoelliptic. Hence we will only be able to construct one semigroup for each element of the family of hypoelliptic operators $-\Delta - \K \partial_t$, where  $\K = \frac{\epsilon + i}{\epsilon^2 + 1}, \epsilon >0$.

Let us consider a space-time domain of the form $\Omega = \ul{\Omega} \times \BR^+ \subset \BR^{n+1}$. Suppose that $\Omega$ is an arbitrary open and complete manifold.

On open and complete manifolds, where completeness is meant with respect to the $L_2-$norm, the Laplacian is essentially self-adjoint on tensors fields with compact support. Applying the regularization procedure that has been described in Subsection \ref{subsec2.2}, we obtain, after using the spectral theorem, the following integral operator
\beqn
\Gamma_t^\K = \int_0^{+\infty} e^{-\frac{t \lambda}{\K}}~ dE_\lambda. \nonumber
\eeqn
For more details about the application of the spectral theorem to the Dirac operator in the context of Clifford analysis, see \cite{DSS}.
For each $\K$ and $t$ fixed the integral operator defined here is well defined in $L_2(\Omega)$. For $u \in L_2(\Omega)$ we have the following properties:
\begin{enumerate}
\item[\textbf{(i)}]
$(-\Delta - \K \partial_t) \Gamma_t^\K u = \Gamma_t^\K (-\Delta - \K \partial_t)u$; \vspace{0,25cm}

\item[\textbf{(ii)}]
the mapping $t \mapsto \Gamma_t^\K u$ is differentiable; \vspace{0,25cm}

\item[\textbf{(iii)}]
$\partial_t \Gamma_t^\K u = (-\Delta - \K \partial_t ) \Gamma_t^\K u$.
\end{enumerate}

These properties follow immediately from differential properties of semigroups and can be found in \cite{Eva} (Subsection 7.4.1).


\subsection{Dissipative property of the regularized operators} \label{subsec1}

Now we want to verify if the elements of the regularized semigroup $\{ \Gamma_t^\K \}_{t \in \BR_0^+}$ are dissipative. This property is very important because it will give us the possibility to obtain results that are essentially needed for solving of initial-value problems. To do that we first prove that, for each fixed  $\K$, the elements of the semigroup satisfy the conditions of Lemma \ref{lem1}, i.e,  $\mathrm{Range}(\mu - (-\Delta)) = X$, for some $\mu > 0$, or equivalent $\mathrm{Range}(\mu - \Delta) = X$, for some $\mu < 0$. First we need to consider the following auxiliar result:

\begin{lemma} \label{lem2}
Suppose that $u \in L_p (\Omega) + L_q(\Omega)$, i.e. $u=u_1+u_2$ with $u_1 \in L_p(\Omega)$, $u_2 \in L_q(\Omega)$ and $1 < p \leq q < 3$. If
\beqn
-\Delta u = \mu u,  \nonumber
\eeqn
for some  $\mu > 0$, then $u$ is identically zero.
\end{lemma}
\begin{proof}
In order to prove this statement we now introduce three auxiliar functions. For $x_0 \in \Omega$ and arbitrary $0<r<s$ we can construct a Lipschitz continuous and almost everywhere differentiable function $\phi_{r,s}$ with the following properties
\begin{enumerate}
\item[\textbf{(i)}] $\ds 0 \leq \phi_{r,s} \leq 1$; \vspace{0,25cm}

\item[\textbf{(ii)}] $\ds \mathrm{supp} \phi_{r,s} \subseteq B_s(x_0) = \{ x \in \Omega : ||x-x_0||_\Omega <s  \}$; \vspace{0,25cm}

\item[\textbf{(iii)}] $\ds \phi_{r,s} = 1,$ on $B_r(x_0)$; \vspace{0,25cm}

\item[\textbf{(iv)}] $\ds \lim_{r,s \rightarrow +\infty} \phi_{r,s} = 1$; \vspace{0,25cm}

\item[\textbf{(v)}] $\ds \left| d \phi_{r,s} (x) \right| = \left| D \phi_{r,s} (x) \right| \leq \frac{c}{s-r}$ almost everywhere. \vspace{0,25cm}
\end{enumerate}

Since $\phi_{r,s}$ is a scalar function, property \textbf{(v)} is a direct consequence of the properties of the Laplace operators acting on a differential form in this more general setting.

We define the following auxiliar function, denoted by $h_1$, as follows
\beqn
h_1(t)=\left \{ \begin{array}{ll}
\ds t^{p-2} &\mbox{ if } t \geq 1 \\
\ds (\gamma + t^2)^{\frac{q-2}{2}} &\mbox{ if } t < 1 - \gamma
\end{array} \right., \nonumber
\eeqn
with $0 < \gamma < 1$ small enough. Hence, for $1<p \leq q < 3$ we have
\beqn
th_1'(t)= \left\{ \begin{array}{ll}
\ds (p-2) t^{p-2} &\mbox{ if } t \geq 1 \\
\ds (q-2)t^2 (\gamma + t^2)^{\frac{q-2}{2}-1} &\mbox{ if } t < 1- \gamma
\end{array} \right., \nonumber
\eeqn
which proves
\beqn
|t h_1'(t)| \leq \eta h_1(t), \label{1*}
\eeqn
for all $t \notin ]1-\gamma,1[$ and some $0 < \eta < 1.$

Since $h_1$ acts outside the interval $]1-\gamma,1[$ we need to consider a third auxiliar function $h_2,$ acting on the interval such that inequality (\ref{1*}) holds  also for $h_2$. These auxiliar functions $h_1$ and $h_2$ are necessary in order to give some control over the regularity of $\phi_{r,s}.$

After these preliminary observations we the proof.  Take an arbitrary element $\phi$ from the family of $\{ \phi_{r,s}\}$. Let us consider the term $\ds \langle  \phi^2 h_1(|u|)u,u\rangle $. For $\mu > 0$ we then have
\beqn
-\mu \langle  \phi^2 h_1(|u|)u,u\rangle  & = & \langle \phi^2 h_1(|u|)u,-\mu u\rangle \nonumber \\& = &  \langle \phi^2 h_1(|u|)u, \Delta u)\rangle \nonumber \\
&  = & \langle \phi^2 h_1(|u|)u, -DD u\rangle  \nonumber \\
& = & \langle D (\phi^2 h_1(|u|)u), D u\rangle \nonumber
\eeqn

Applying the chain rule and the Leibniz rule, the last expression then turns out to be equal to
\beqn
 \langle \phi^2 h_1(|u|)~Du, Du \rangle  + \langle \phi^2 h_1'(|u|)(uDu), Du \rangle + 2\langle \phi h_1(|u|)~uD\phi, D u \rangle, \nonumber
\eeqn
where $Du,~D\phi$ are vectorial expressions. Hence, we get
\begin{align}
& -\mu \langle  \phi^2 h_1(|u|)u,u\rangle =  \nonumber \\
&\langle \phi^2 h_1(|u|)~Du, Du \rangle  + \langle \phi^2 h_1'(|u|)(uDu), Du \rangle + 2\langle \phi h_1(|u|)~uD\phi, D u \rangle.\label{6*}
\end{align}

Since $\mu > 0$ we have that $\ds -\mu \langle  \phi^2 h_1(|u|)u,u\rangle < 0$. Consequently,
\beqn
\langle \phi^2 h_1(|u|)~Du, Du \rangle  + \langle \phi^2 h_1'(|u|)(uDu), Du \rangle & \leq &  - 2\langle \phi h_1(|u|)~uD\phi, D u \rangle \nonumber \\
& \leq & \left| - 2\langle \phi h_1(|u|)~uD\phi, D u \rangle \right| \nonumber \\
\label{AB} & \leq & 2\left| \langle \phi h_1(|u|)~uD\phi, D u \rangle \right|,
\eeqn

By $|t h_1'(t)| \leq \eta h(t)$, it follows that
\beqn
-\eta \langle \phi^2 h_1(|u|) Du,Du\rangle  & \leq & - \langle \phi^2 |u| h_1'(|u|)~ Du,Du\rangle  \nonumber \\
& = & - 2^n \int_\Omega \phi^2 h_1'(|u|)~|u|~|Du|^2~dx~dt \label{3*}
\eeqn From this property we get
\beqn
0 \leq (1-\eta)\langle \phi^2 h_1(|u|)~Du,Du\rangle \leq \langle \phi^2 h_1(|u|)~Du,Du\rangle + \langle \phi^2 h_1'(|u|)~|u|~|Du|,  |Du| \rangle \nonumber
\eeqn

The above established estimate together with (\ref{AB}) implies that
\begin{eqnarray}
0 \leq (1-\eta)\langle \phi^2 h_1(|u|)~Du,Du\rangle \leq 2\left| \langle \phi h_1(|u|)~uD\phi, D u \rangle \right| \label{4*}
\end{eqnarray}

Applying Schwarz's inequality to the right-hand side of (\ref{4*}) leads to
\begin{align*}
&  2 |\langle \phi h_1(|u|)~uD\phi, D u \rangle | \nonumber \\
& \leq 2 \left( 2^n \int_\Omega \phi^2 h_1(|u|)~|Du|^2~|D\phi|^2~dx~dt \right)^{\frac{1}{2}} \left( 2^n \int_{\Omega} \phi^2 h_1(|u|)~|u|^2~dx~dt \right)^{\frac{1}{2}} \nonumber \\
& \leq 2 ||D\phi||_\infty \left( 2^n \int_\Omega \phi^2 h_1(|u|)~|Du|^2~dx~dt \right)^{\frac{1}{2}} \left( 2^n \int_{\mathrm{supp} \phi} h_1(|u|)~|u|^2~dx~dt \right)^{\frac{1}{2}}, \nonumber
\end{align*}
where $\ds ||D\phi||_\infty = \sup_{x \in \mathrm{supp}\phi} |(D\phi)(x)|.$

Hence, from (\ref{4*}) we may further conclude that
{\small
\begin{align}
&(1 - \eta) ~ 2^n \int_\Omega \phi^2 h_1(|u|)~|Du|^2~dx~dt \nonumber \\
&\leq 2 ~ ||D\phi||_\infty \left( 2^n \int_\Omega \phi^2 h_1(|u|)~|Du|^2~dx~dt \right)^{\frac{1}{2}} \left( 2^n \int_{\mathrm{supp} \phi} h_1(|u|)~|u|^2~dx~dt \right)^{\frac{1}{2}} \label{12*}
\end{align}}

Squaring both sides of (\ref{12*}) and dividing them after that by $\ds (1-\eta)^2 ~ 2^n  \int_\Omega \phi^2 h_1(|u|)~|Du|^2~dx~dt$ leads to
{\small\beqn
\int_\Omega \phi^2h_1(|u|)~|Du|^2 dx~dt \leq 4 (1-\eta)^{-2} ~ ||D \phi||_\infty^2 \int_{\mathrm{supp} \phi} h_1(|u|)|u|^2~dx~dt. \label{5*}
\eeqn}

For $\gamma \rightarrow 0^+,$ the expression $ h_1(|u|) |u|^2$ converges to
\beqn
h(|u|)~|u|^2 =
\left \{ \begin{array}{ll}
|u|^p & \mbox{ if } |u| \geq 1 \\
|u|^q & \mbox{ if } |u| < 1
\end{array} \right. \nonumber
\eeqn
This expression is globally integrable whenever
\beqn
u \in L_p(\Omega) + L_q(\Omega). \nonumber
\eeqn

Now it remains to prove that under these conditions $u \equiv 0$. If $s \rightarrow +\infty$, then  $\mathrm{supp}\phi \rightarrow \Omega$. Hence, for the right-hand side of (\ref{5*}) we obtain
\beqn
\lim_{s \rightarrow +\infty}  \lim_{\gamma \rightarrow 0^+}\int_{\mathrm{supp} \phi} h_1(|u|)|u|^2~dx~dt = \int_\Omega h(|u|)~|u|^2~dx~dt. \nonumber
\eeqn
This limit is finite as a consequence of the preceding considerations.

Finally, by property \textbf{(v)} $\phi$, $||D \phi||_\infty \rightarrow  0$, if $s$  tends to infinity. Hence,
\beqn
\int_\Omega h(|u|)~|Du|^2~dx~dt = 0, \nonumber
\eeqn
i.e., $Du=0$. This fact implies that $-\Delta u = 0$. Finally, we arrive at $ u=\mu^{-1} \Delta u =0.$
\end{proof}

Under these conditions we immediately obtain the main result of this subsection.
\begin{lemma}\label{lem3}
$-\Delta$ is dissipative on $L_p^0(\Omega),$ for $1 < p < 3.$
\end{lemma}
\begin{proof}
If $u \in L_p^0(\Omega)$, then
\beqn
\langle |u|^{p-2} u, -\Delta u \rangle & = & \langle D(|u|^{p-2}u),Du \rangle \nonumber \\
& = & \langle |u|^{p-2}~Du,Du \rangle + (p-2)\langle |u|^{p-3} (uDu), Du \rangle. \nonumber
\eeqn

We have
\beqn
0 \leq \left| \langle |u|^{p-3} (uDu), Du \rangle\right| & \leq & 2^n \int_\Omega |u|^{p-3}~|u|~|Du|~|Du|~dx~dt \nonumber \\
& = & \langle |u|^{p-2}~Du, Du \rangle, \nonumber
\eeqn
i.e., with $|p-2|<1$
\beqn
\langle |u|^{p-2}u, -\Delta u \rangle \leq 0. \nonumber
\eeqn
\end{proof}


\subsection{Main result} \label{subsec2}

The aim of this subsection is to determine for which values of $p$ the property $u \in L_p(\Omega)$ implies the uniqueness of the associated semigroup $\{ \Gamma_t^\K \}_{t \in \BR_0^+}$ and that $\Gamma_t^\K u$ is a solution of the regularized Schr\"odinger equation.

\begin{theorem} \label{teo1}
Let $\{ \Gamma_t^\K \}_{t \in \BR_0^+}$ be the regularized Schr\"odinger semigroup acting on $L_2(\Omega)$. Then  $\ds ||\Gamma_t^\K u||_p \leq ||u||_p,$ for all $u \in L_p(\Omega) \cap L_2(\Omega)$ and $\ds \frac{3}{2} < p < 3.$

Therefore, $\{\Gamma_t^\K \}_{t \in \BR_0^+}$ extends to a contraction semigroup on $L_p^0(\Omega)$ for $\ds \frac{3}{2} < p < 3$. Moreover, $\Gamma_t^\K u$ satisfies the regularized  Schr\"odinger equation
\begin{eqnarray*}
\K \partial_t (\Gamma_t^\K u) = - \Delta (\Gamma_t^\K u),
\end{eqnarray*}
for $u \in L_p(\Omega)$ and $\{ \Gamma_t^\K \}_{t \in \BR_0^+}$ is unique.
\end{theorem}
\begin{proof}
The closure $A$ of $\ds -\Delta|_{L_p^0(\Omega)}$ in $L_p(\Omega)$ is dissipative for $1<p<3$.

Furthermore, $\mu - A$ is surjective for $\mu>0$ and for $1<p<3$. In fact, if this was wrong, then there would exist a $u \in L_{p'}(\Omega)$ such that $\ds \langle u, (\mu - A ) v \rangle = 0,$ for all $v \in L_p^0(\Omega).$ This would imply $\Delta u = -\mu u,$ for $\mu > 0,$ establishing a contradiction to Lemma~\ref{lem1}.

From $p' < 3$ we get the restriction $\ds p > \frac{3}{2}$. Hence,  $A$ generates a contraction semigroup $\ds \{ Q_t \}_{t \in \BR_0^+}$ for $\ds \frac{3}{2} < p < 3.$

Next, we show that the semigroups $Q_t$ and $\Gamma_t^\K$ agree on
\beqn
L_2 \cap L_p = L_2(\Omega) \cap L_p(\Omega). \nonumber
\eeqn

For this it is sufficient to show that $(\mu - (- \Delta))^{-1}$ and $(\mu - A)^{-1}$ coincide on $L_2 \cap L_p$. Suppose that $u \in L_2 \cap L_p$, $(\mu - (-\Delta))^{-1}u=v,$ $(\mu - A)^{-1}u = w.$ Then $v \in L_2,$ $w \in L_p,$ $v-w \in L_2 + L_p$ and $\Delta(v-w)=-\mu(v-w),$ $\mu>0.$ According to Lemma \ref{lem2}, we have $v=w,$ $\ds \{ Q_t \}_{t \in \BR_0^+} = \{\Gamma_t^\K  \}_{t \in \BR_0^+}$ on $L_2 \cap L_p.$

This proves the estimate $\ds ||\Gamma_t^\K u||_p \leq ||u||_p$, for $\ds \frac{3}{2} < p < 3.$

Since $\Gamma_t^\K u$ satisfies the regularized  Schr\"odinger equation for $u \in D_\Delta$ and since this domain is dense in $\ds L_p(\Omega),$ $\Gamma_t^\K u$ also satisfies the regularized   Schr\"odinger equation, but at the first instance only in distributional sense. The hypoellipticity of the regularized  Schr\"odinger operator implies this property in the pointwise sense only.

Now, we prove the uniqueness. Suppose that $A'$ is the infinitesimal generator of another contraction semigroup $\ds \{ P_t\}_{t \in \BR_0^+}$, such that $P_t u$ satisfies the regularized  Schr\"odinger equation. Then we have to show $(\mu - A')^{-1} = (\mu - (-\Delta))^{-1}.$

We have $(\mu - A')^{-1} u = v$ which means $v \in D_A,$ and $(\mu - A')v = u.$ If $v \in D_A,$ then
$$\begin{array}{rcl}
\ds t^{-1}(P_t v - v) & \rightarrow & L'v \in L_p(\Omega), \\
& & \\
\ds t^{-1}(P_{s+t} v - P_s v) & \rightarrow  & P_s A'v \in L_p(\Omega),
\end{array}$$
for any fixed $s>0$. $P_t u$ satisfies the regularized   Schr\"odinger equation. Therefore,
\beqn
t^{-1}(P_{s+t}v - P_s v) \rightarrow \partial_s P_s v = -\Delta P_s v, \nonumber
\eeqn
i.e., $P_s A' v = -\Delta P_s v$. Then
\beqn
A' v = \lim_{s \rightarrow 0} (-\Delta P_s v) = -\Delta v \nonumber
\eeqn
in the distributional sense. It follows that $v \in L_p(\Omega)$ satisfies $(\mu - (-\Delta))v=u.$ On the other hand, if $(\mu - (-\Delta))^{-1} u = w$, then $w \in L_p(\Omega)$ and
$$\begin{array}{rcl}
(\mu - (-\Delta)) w & = &  u, \\
& & \\
\Delta (v-w) & = & - \mu (v-w), \quad \mu>0.
\end{array}$$

According to Lemma \ref{lem2} we may conclude that $v=w$. This establishes our result.
\end{proof}


\subsection{The flat oriented torus case}

In this subsection and the following one we want to give for some very special examples of manifolds explicit analytic representation formulas for the fundamental solution to the Schr\"odinger operator.

In this subsection we present some explicit formulas for the solutions to the Schr\"odinger equation on conformally flat $n$-tori (and conformally flat $k$-cylinders). Conformally flat means that these manifolds have a vanishing Weyl tensor. This property is equivalent to the fact that the manifold possesses an atlas whose transition functions are conformal maps in the sense of Gauss (which are holomorphic functions in dimension $n=2$ and M\"obius transformations for $n \ge 3$). So, in the case $n=2$ the set of conformally flat manifolds coincides with the set of holomorphic Riemann surfaces.
\par\medskip\par

As is well known, we obtain conformally flat $n$-tori by forming the quotient of $\mathbb{R}^n$ with an $n$-dimensional torsion free lattice
$$
\Omega_n:= \mathbb{Z} v_1 + \cdots + \mathbb{Z} v_n
$$
where the elements $v_i$ $(i=1,...n)$ are chosen in a way that they are $\mathbb{R}$-linearly independent vectors from $\mathbb{R}^n$. Each element of the lattice $\Omega_n$ then can be written in the form
$$
v = m_1 v_1 + \cdots + m_n v_n
$$
with integers $m_1,...,m_n \in \mathbb{Z}$.
\par\medskip\par
Now let $U \subset \mathbb{R}^n$ be an open set.
A function $f:U \times \mathbb{R}^+ \to \BC_n$ that satisfies $f(x+v,t) = f(x,t)$ for all $v \in \Omega_n$ then naturally descends to the $n$-dimensional torus $T_n(v_1,...,v_n) := \mathbb{R}^n/\Omega_n$ by forming $f':=p(f)$, where $p: \mathbb{R}^n \to T_n$, $x \mapsto x \mod \Omega_n$ is the canonical projection from the spatial part $\mathbb{R}^n$ down to the manifold $T_n(v_1,...,v_n)$. For the sake of simplicity we shall write $T_n$ instead of $T_n(v_1,...,v_n)$ when it is clear which basis vectors $v_1,...,v_n$ are considered. Notice that this projection $p$ leaves the time variable $t$ invariant; it only acts on the spatial variables.
\par\medskip\par
Following \cite{Kuiper} and others, the manifolds $T_n$ are actually all conformally flat.
Next, following e.g. \cite{KraRyan2}, the decomposition of the lattice $\Omega_n$ into the direct sum of the sublattices $\Omega_l:= \mathbb{Z} v_1 + \cdots + \mathbb{Z} v_l$ and $\Omega_{n-l}:= \mathbb{Z} v_{l+1} + \cdots + \mathbb{Z} v_n$ gives rise to conformally inequivalent different spinor bundles, denoted by $E^{(q)}$, on $T_n$ by making the identification $(x,X) \Longrightarrow (x+\underline{m}+\underline{n}),(-1)^{m_1+\cdots+m_l}X)$ with $x \in \mathbb{R}^n,X \in \BC_n$. Since $T_n$ is orientable, we are dealing with examples of spin manifolds in this context here.
\par\medskip\par
Notice that the different spin structures on a spin manifold $M$ are detected by the number of distinct homomorphisms from the fundamental group $\Pi_{1}(M)$ to the group ${\mathbb{Z}}_{2}$. In the case of the $n$-torus we have that $\Pi_{1}(T_n)={\mathbb{Z}}^{n}$. There are two homomorphisms of $\mathbb{Z}$ to $\mathbb{Z}_{2}$. The first is $\theta_{1}:{\mathbb{Z}}\rightarrow {\mathbb{Z}}_{2}:\theta_{1}(n) \equiv 0 \mod 2$ while the second is the homomorphism $\theta_{2}:{\mathbb{Z}}\rightarrow {\mathbb{Z}}_{2}:\theta_{2}(n) \equiv 1 \mod 2$. Consequently, there are $2^{n}$ distinct spin structures on $T_{n}$. $T_{n}$ is also an example of a Bieberbach manifold. Further details of spin structures on the $n$-torus and other Bieberbach manifolds can be found in \cite{f,mp,pf}.
\par\medskip\par
By applying the projection map $p$ to the regularized Schr\"odinger operator $(\Delta-{\bf k} \partial_t)$, we induce a second order operator $(\Delta'-{\bf k} \partial_t)$ on the spin manifolds $T_n \times \mathbb{R}^+$, which then is the regularized Schr\"odinger operator on this spin manifold.
\par\medskip\par
To construct the fundamental solution of the associated toroidal Schr\"odinger operator we periodize the fundamental solution
$$
e^{\epsilon}_- (x,t) := (\epsilon+i)\frac{H(t)}{(4 \pi(\epsilon+i)t)^{n/2}} \exp\left(-\frac{(\epsilon+i)|x|^2}{4(\epsilon^2+1)t}\right), \quad \epsilon> 0
$$
of the hypoelliptic operator $(\Delta-{\bf k} \partial_t)$ over the period lattice. More precisely, this is achieved by forming the sum
$$
\wp^{\epsilon}_{q}(x,t) := \sum_{\underline{m} \in \Omega_l} \sum\limits_{\underline{n} \in \Omega_{n-l}}(-1)^{m_1+\cdots+m_l} e^{\epsilon}_{-}(x+\underline{m}+\underline{n};t)
$$
in which we take care of the proper minus sign that appears in the construction of the particular spinor bundle $E^{(q)}$ that we consider. The normal convergence of this series in $\mathbb{R}^n\backslash \Omega_n$ has been proved previously in \cite{KV1} to which we refer the reader for the technical details. The projection $p(\wp^{\epsilon}_{q}(x,t))$ thus descends to a well-defined spinor section $P^{\epsilon}_q$ that is in the kernel of the toroidal Schr\"odinger operator acting on the chosen  spinor bundle $E^{(q)}$ of the conformally flat torus $T_n$. As one can easily verify, see again \cite{KV1} for details, these spinor sections then are the fundamental solutions to the associated regularized Schr\"odinger operator on these manifolds. This is because they serve as the Green's kernel to the toroidal Schr\"odinger operator reproducing all spinors in the kernel of this operator on these manifolds.
\par\medskip\par
We can say much more. We can also construct every spinorial solution to the Schr\"odinger operator on $T_n$ as an additive series over linear combinations of the section $P^{\epsilon}_q$ and its partial derivatives. One gets uniqueness up to an entire real-analytic function that only depends on the time variable $t$. More precisely, adapting from  \cite{KV1}, we can directly establish that
\begin{theorem}
Let $S \subset \mathbb{R}^{n} \times \mathbb{R}^+$ be a closed subset that has the property that $S+v = S$ for all $v \in \Omega_n$.
Let $a_1,\ldots,a_p \in \mathbb{R}^{n+1} \backslash S$ be a finite set of points that are incongruent modulo $V$. Suppose
that $u:T_n \times \mathbb{R}^+ \backslash \{a_1,...,a_p\} \mapsto E^{(q)}$ is a spinor section of the regularized Schr\"odinger operator acting on the spinor bundle of $E^{(q)}$  which has at most singularities at the points of $a_i$ of order $K_i$. Then there exist constants $b_1,...,b_p \in \BC_n$ and a real analytic function $\phi=\phi(t)$ such that
$$
u(x,t) = p\left( \sum_{i=1}^p \sum_{m=0}^{K_i-(n-1)} \sum_{m=m_1+...+m_n} \Big[\wp^{\epsilon}_{m_1,...,m_n;q}(x-a_i,t) b_i\Big] + \phi(t)\right),
$$
where $\wp^{\epsilon}_{m_1,...,m_n;q}(x-a_i,t) = \frac{\partial^{m_1+...+m_n}}{\partial x_1^{m_1} \cdots \partial x_n^{m_n}}\wp^{\epsilon}_{q}(x-a_i,t)$.
\end{theorem}
By means of the section $P^{\epsilon}_q$ we can also obtain the fundamental to the original Schr\"odinger equation in the limit case $\epsilon \to 0$ on the torus $T_n$ with values in the spinor bundle $E^{(q)}$. In \cite{KV1} we have shown
\begin{theorem} Let $V' \subset T_n$ be a domain.
For all $1 \le p < +\infty$ we have the following weak convergence in $W_p^{-n/2-1}(V')$,
$$
\langle P^{\epsilon}_q,\phi \rangle \to \langle P_q,\phi\rangle,\; \phi \in W_p^{n/2+1}(V')
$$
when $\varepsilon \to 0^+$.
\end{theorem}
\begin{remark}
The toroidal case is a very special case in the more general context of this paper here. First of all, these tori have the property that they can be constructed by factoring $\mathbb{R}^n$ by a discrete Kleinian group under whose action the regularized Schr\"odinger operator is totally left invariant. Notice that (up to conjugation) only translation subgroups of the $SO(n)$ have the property that they leave the set of null solutions to the Schr\"odinger totally invariant. Furthermore, it is due to the discreteness of the group, that we can describe the solutions of the Schr\"odinger operator as discrete {\it additive series}. In the more general context discussed in the other parts of this paper we cannot expect the fundamental solutions to be expressible in terms of additive periodizations of the fundamental solution to the regularized Schr\"odinger operator in $\mathbb{R}^n \times \mathbb{R}^+$. Of course, we also obtain similar series representations in the context of conformally cylinders that are constructed by factoring $\mathbb{R}^n$ by a $k$-dimensional sublattice of $\Omega_k$ for $k=1,...,n-1$. In this case the fundamental solution is simply a subseries of $P^{\epsilon}_q$ in which one only sums over the lattice points that belong to the sublattice $\Omega_k$.
\end{remark}

\par\medskip\par

A further speciality of the torus case (and also the cylinder cases) is the orientability of this manifolds which makes $T_n$ to a spin manifold. In the following subsection we explain how we can adapt the formulas that we presented in this subsection to non-orientable counterparts of the manifolds considered here.

\subsection{A class of non-orientable conformally flat manifolds}

The oriented cylinder $C$ defined as the topological quotient $\mathbb{R}^2/\mathbb{Z}$ has a natural non-oriented counterpart, namely the M\"obius strip. Also the torus $T_2:=\mathbb{R}^2/\mathbb{Z}^2$ has such a counterpart, namely the Klein bottle. In both cases we can construct these manifolds by gluing the same vertices of the fundamental domain of the associated one-dimensional resp. two-dimensional translation group (that lead to the cylinder resp. torus) together, both with opposite orientation, which however destroys the orientability.
\par\medskip\par
In the $n$-dimensional setting we can construct a family of non-oriented analogues of these manifolds from the oriented $k$-cylinders defined by $C_k := \mathbb{R}^n/\Omega_k$ where $k \in \{1,...,n-1\}$ and where $\Omega_k \subset \mathbb{R}^k$ is a $k$ dimensional  lattice spanned by $k$ $\mathbb{R}$-linearly independent vectors $\underline{v}_1,...\underline{v}_k \in \mathbb{R}^k$.
\par\medskip\par
Let $\underline{x}$ be a reduced vector from $\mathbb{R}^k$. Suppose that $\underline{v}:= m_1 \underline{v}_1 + \cdots + m_k \underline{v}_k$ is a vector from that lattice $\Omega_k \subset \mathbb{R}^k$.
\par\medskip\par
{\bf 3.5.1. Higher dimensional M\"obius strips}
\par\medskip\par
Similar to the classical case in three dimensions one can introduce higher dimensional analogoues of the M\"obius strip by the factorization
$$
{\cal{M}}_k^{-} = \mathbb{R}^n/\sim
$$
where $\sim$ is now defined by the map $$(\underline{x}+\underline{v},x_{k+1},...,x_{n-1},x_n) \mapsto (x_1,...,x_k,x_{k+1},...,x_{n-1},{\rm sgn}(\underline{v})x_n).$$
Here, for $\underline{v} = m_1 \underline{v}_1 + \cdots + m_k\underline{v}_k$ we write
${\rm sgn}(\underline{v}) = \left\{ \begin{array}{cc} 1 & {\rm if}\; \underline{v} \in 2 \Omega_k \\ -1 & {\rm if}\; \underline{v} \in \Omega_k \backslash 2 \Omega_k. \end{array} \right. $

We recognize the classical M\"obius strip in the case $n=2,k=1$ in which the pair $(x_1+v_1,x_2,X)$ is mapped to $(x_1,-x_2,X)$ after one period.

Due to the switch of the minus sign in the $x_n$-component we indeed deal here with non-orientable manifolds, so ${\cal{M}}_k^-$ are not spin manifolds anymore.

\par\medskip\par

We can say more. Analogously, to the case of a spin manifold we can set up several distinct pin bundles (associated to the $Pin(n)$ group instead to the spin group $Spin(n)$), namely by mapping for instance the tupel
$$(\underline{x}+\underline{v},x_{k+1},...,x_n,X) \; {\rm to}\; (\underline{x}, x_{k+1},...,x_{n-1},{\rm sgn}(v) x_n,(-1)^{m_1+\cdots+m_k}X).$$
For simplicity let us first explain the construction for the trivial pin bundle of the manifold ${\cal{M}}_k^-$ where the tupel  $$(\underline{x}+\underline{v},x_{k+1},...,x_{n-1},x_n,X)$$ is mapped to $$(\underline{x},x_{k+1},...,x_{n-1},{\rm sgn}(\underline{v})x_n,X).$$
Now we can use the same periodization argument as used for the oriented $k$-cylinders in the previous subsection, in order to obtain an explicit formula for the fundamental solution of the regularized hypoelliptic Schr\"odinger operator on the non-oriented manifolds ${\cal{M}}_k^-$. However, instead of applying the ``symmetric'' periodization over the period lattice we have to apply the ``anti-symmetric'' periodization, induced by $\sim$.
\par\medskip\par
Again, let $e^{\epsilon}_{-}(x,t)$ be the fundamental solution to the hypoelliptic regularized Schr\"odinger operator $\Delta-{\bf k} \partial_t$ in Minkowski space-time. Then we may obtain the fundamental solution on the manifold ${\cal{M}}_k^- \times \mathbb{R}^+$ associated with the trivial bundle by the series
$$
P(x,t):=p_-\left( \sum_{\underline{v} \in \Omega_k}  e^{\epsilon}_{-}(\underline{x}+\underline{v},x_{k+1},...,x_{n-1},{\rm sgn}(v)x_n;t)\right)
$$
where $p_-$ now stands the canonical projection from $\mathbb{R}^k \to {\cal{M}}_k^- = \mathbb{R}^k/\sim$.
\par\medskip\par
Notice that each term $e^{\epsilon}_{-}(\underline{x}+\underline{v},x_{k+1},...,x_{n-1},{\rm sgn}(v)x_n;t)$ of the appearing series  actually is annihilated by the regularized hypoelliptic Schr\"odinger operator $\Delta-{\bf k} \partial_t$.

If a function $f(x_1,...,x_{n-1},x_n)$ is annihilated by the Laplacian $\Delta:=\sum_{i=1}^n \frac{\partial^2 }{\partial x_i^2}$ in the vector variable $(x_1,...,x_n)$, then the function $$g(x_1,...,x_{n-1},x_n):=f(x_1,...,-x_{n-1},x_n)$$ again turns out to be harmonic with respect to the same vector variable $(x_1,...,x_n)$. Since the Laplacian differentiates twice each variable, the minus sign is compensated after the second derivation in the $x_n$-direction. Since the minus sign change only occurs in a spatial variable, it has no influence on the variable $t$. Since the series is per construction invariant under $\sim$, it descends to a well-defined section on the manifold ${\cal{M}}_k^-$. On the manifold it is then the fundamental solution to the associated hypoelliptic regularized Schr\"odinger operator.
\par\medskip\par
In the cases of the other pin bundles that we mentioned, we need to add the corresponding minus sign in the sum in front of the anti-multiperiodic expression $P$. More precisely, the corresponding fundamental solution then is given by
$$
P^{(q)}(x,t):=p_-\left( \sum_{\underline{v} \in \Omega_l \oplus \Omega_{k-l}} (-1)^{m_1+\cdots+m_l} e^{\epsilon}_{-}(\underline{x}+\underline{v},x_{k+1},...,x_{n-1},{\rm sgn}(v)x_n;t)\right).
$$
{\bf 3.5.2. Higher dimensional generalizations of the Klein bottle}.
\par\medskip\par
Finally, we turn to discuss higher dimensional generalizations of the Klein bottle. To leave it simple we consider an $n$-dimensional normalized lattice of the form $\Omega_n:=\Omega_{n-1} + \mathbb{Z} e_n$ where $\Omega_{n-1} \subset \mathbb{R}^{n-1}$. Notice that every arbitrary $n$-dimensional lattice can be transformed into a lattice of the latter form by simply applying a rotation and a dilation.

Now we may introduce higher dimensional generalization of the classical Klein bottle by the factorization
$$
{\cal{K}}_n:=\mathbb{R}^n/\sim^*
$$
where $\sim^*$ is now defined by the map
$$
(\underline{x} + \sum_{i=1}^{n-1} m_i \underline{v}_i +(x_n + m_n) e_n) \mapsto(x_1,\cdots,x_{n-1},(-1)^{m_n} x_n).
$$
Alternatively, these manifolds can be constructed by gluing finitely many conformally flat manifolds together, which is according to \cite{Ry2003} another argument for being conformally flat.
Here, and in the remaining part of this subsection, $\underline{x}$ denotes a shortened vector in $\mathbb{R}^{n-1}$. In the case $n=2$ we obtain the classical Klein bottle. Notice that in contrast to the M\"obius strips, in this context here the minus sign switch occurs in one of the component on which the period lattice acts, too. As for the M\"obius strips we can again set up distinct pin bundles. By decomposing the complete $n$-dimensional lattice $\Omega_n$ into a direct sum of two sublattices $\Omega_n = \Omega_l \oplus \Omega_{n-l}$ we can again construct $2^n$ distinct pin bundles by considering the maps
$$
(\underline{x} + \sum_{i=1}^{n-1} m_i \underline{v}_i, x_n + m_n,X) \mapsto(x_1,\cdots,x_{n-1},(-1)^{m_n} x_n,(-1)^{m_1+\cdots+m_l}X).
$$
By similar arguments as before we can express the fundamental solution of the regularized Schr\"odinger operator  on the manifold ${\cal{K}}_n \times \mathbb{R}^+$ associated with values in that pin bundle by the series
$$
P(x,t):=p_*\left( \sum_{(\underline{v},m_n) \in \Omega_{n-1}\times \mathbb{Z}} (-1)^{m_1+\cdots+m_l}
e^{\epsilon}_{-}(\underline{x}+ \underline{v}+((-1)^{m_n} x_n + m_n) e_n;t)\right)
$$
where $p_*$ now stands the canonical projection from $\mathbb{R}^n \to {\cal{K}}_n = \mathbb{R}^n/\sim^*$. Again, for the trivial bundle the parity factor $(-1)^{m_1+\cdots+m_l}$ simplifies to $+1$.
\par\medskip\par
Each term $e^{\epsilon}_{-}(\underline{x}+ \underline{v}+((-1)^{m_n} x_n + m_n) e_n;t)$ of this series actually is in the kernel of the regularized hypoelliptic Schr\"odinger operator $\Delta-{\bf k} \partial_t$, for the same reason as for the M\"obius strip.
\par\medskip\par
{\bf Final remark}.
As in the cases treated in the previous section we can also obtain from these formulas a fundamental solution to the Schr\"odinger operator in the limit case $\epsilon \to 0$. To do so one has to apply the same procedure as explained previously, so we leave this as an exercise to the reader.

\subsection{The Laplacian for non-flat manifolds}\label{subsec3}

Finally, in this subsection we want to briefly outline how we can deal with non-flat arbitrary Minkowski manifolds.
\par\medskip\par
The classical Laplace operator is not suited for an arbitrary Minkowski manifold, since it fails to take into consideration its underlined geometric structure, e.g. its curvature or its non-Riemannian metric. Hence, we aim now to outline how we can extend some of the previous results to a Schr\"odinger-type operator  where the Laplace operator is replaced by the Bochner-Laplacian or by the G\"unter-Laplacian. For that, we shall write these equations in local cartesian coordinates and associated differential forms rather than using intrinsic metric tensor coordinates.

Differential forms have the advantage of fit naturally into integral formulation, since they provide immediate linkage between local and global geometry (topology) simplifying the arising expressions.

If we consider  an $(n+1)-$dimensional arbitrary and complete Minkowski manifold, say $(M,g)$, then in the case of the Bochner-Laplacian we need to impose that $\mathrm{Ric} > 0$ while for the G\"unter-Laplacian we require that  $2R^2 - \mathcal{G}R >0$. With these additional conditions and taking into account (\ref{c1.5*}) and (\ref{c1.4*}), we can establish analogous proofs to the previous results and may conclude that
\begin{itemize}
\item The operators $-\Delta_B$ and $-\Delta_G$ with domain $L_p^0(\bigwedge_k M) $  are dissipative for $1 < p < 3.$

\item $\ds ||\Gamma_t^\K u||_p \leq ||u||_p,$ for all $u \in L_p(\bigwedge_k M) \cap L_2(\bigwedge_k M)$ and $\ds \frac{3}{2} < p < 3$ and, therefore, $\{ \hat{ \Gamma }_t^\K  \}_{t \in \BR_0^+}$ extends to a contraction semigroup on $L_p^0(\bigwedge_k M)$ for $\ds \frac{3}{2} < p < 3$.
\end{itemize}


\section{The regularized Schr\"odinger problem}\label{sec4}

In this section, we show how the semigroup $\{ \Gamma_t^\K \}_{t \in \BR_0^+}$ is related to the regularized Schr\"odinger problem with initial condition. As a consequence of Theorem \ref{teo1} we immediately obtain
\begin{theorem}\label{teo2}
The initial value problem
\beqn
\left\{
\begin{array}{ll}
(-\Delta - \K \partial_t) v = 0, & \quad \mbox{ on } \Omega  \\
v(x,0)=u_0(x), & \quad \mbox{ on } \ul{\Omega}
\end{array}
\right.
\label{7*}
\eeqn
is solvable, with $v(\cdot, t) \in L_p(\Omega)$, whenever $u_0 \in L_p(\Omega)$ and  $\ds \frac{3}{2} < p < 3.$
\end{theorem}

The remaining open question of uniqueness is answered in the following statement
\begin{theorem}\label{teo3}
Let $v=v(x,t)$ be a solution of the regularized Schr\"odinger equation with $v(\cdot,t) \in L_p(\Omega)$ and  $\ds \frac{3}{2} < p < 3.$  Assume further that $||v(\cdot,t)||_p \leq a e^{-|\K| b t}$. Then there exists a uniquely determined function $u_0 \in L_p(\Omega)$, such that $v = \Gamma_t^\K u_0.$
\end{theorem}
\begin{proof}
In the following proof we denote the corresponding space solution by $L_p.$

If $\ds u_0 = \lim_{t_j \rightarrow 0} v(\cdot,t_j)$ in the weak star topology, $u=v-\Gamma_t^\K u_0$, then
\beqn
||u(\cdot,t)||_p \leq ae^{-|\K| bt} \label{8*}
\eeqn
and
\beqn
u(\cdot, t_j) \rightarrow 0, \quad \mbox{ when } t_j \rightarrow 0 \label{9*}
\eeqn
in the distributional sense.

Furthermore, $u$ satisfies the regularized  Schr\"odinger equation since each term does. We have to show that $u=0.$ To do this we consider the Laplace transform of $u$
\beqn
w_\lambda^\K (x) = \int_0^{+\infty} e^{-\frac{t \lambda}{|\K|}} u(x,t) ~ dt. \nonumber
\eeqn

According to (\ref{8*}) the integral converges absolutely for sufficiently large values of $|\K|\lambda$ and for almost all admissible  $x$. Moreover, $w_\lambda^\K \in L_p.$ Next we show that $\Delta w_\lambda^\K = -\K \lambda w_\lambda^\K$ holds in the distributional sense. For any $\psi \in L_p^0(\Omega)$
\beqn
\langle \psi, \Delta w_\lambda^\K \rangle & = &  \langle \Delta \psi, w_\lambda^\K \rangle \nonumber  \\
& = & \int_0^{+ \infty} e^{-\frac{t \lambda}{|\K|}} \langle \Delta \psi , u(\cdot, t) \rangle~ dt. \label{10*}
\eeqn

According to (\ref{8*}) the previous double integral converges absolutely for large $|\K| \lambda$. Using the regularized  Schr\"odinger equation
\beqn
\langle \Delta \psi , u(\cdot, t)\rangle = -\K \partial_t \langle \psi, u(\cdot, t)\rangle, \nonumber
\eeqn

we obtain via integration by parts
\begin{align*}
\langle \psi, \Delta w_\lambda^\K \rangle &= - \int_0^{+\infty} e^{-\frac{t \lambda}{|\K|}} \partial_t \langle \psi, u(\cdot,t) \rangle ~ dt \\
&= - \lim_{\begin{array}{c} t_j \rightarrow 0 \\  N \rightarrow +\infty \end{array}} \int_{t_j}^N e^{-\frac{t \lambda}{|\K|}} \partial_t \langle \psi, u(\cdot,t) \rangle ~ dt \\
&= - \lim_{\begin{array}{c}t_j \rightarrow 0 \\  N \rightarrow +\infty \end{array}} \left[ \lambda \int_{t_j}^N e^{-\frac{t \lambda}{|\K|}} \langle \psi, u(\cdot,t) \rangle ~ dt  + e^{- \frac{N \lambda}{|\K|}} \langle \psi, u(\cdot,N)\rangle \right. \\
& \left. \qquad \qquad \qquad \qquad \qquad \qquad \qquad -  e^{-\frac{t_j \lambda}{|\K|}} \langle \psi, u(\cdot,t_j)\rangle\right] \\
&= - \lambda \int_0^{+\infty} e^{-\frac{t \lambda}{|\K|}} \langle \psi, u(\cdot,t) \rangle ~ dt
\end{align*}
since $e^{- \frac{N \lambda}{|\K|}} \langle \psi, u(\cdot,N)\rangle$ by (\ref{8*}) and $e^{- \frac{t_j \lambda}{|\K|}} \langle \psi, u(\cdot,t_j)\rangle$ by (\ref{9*}).

Summarizing, we arrive at $\Delta w_\lambda^\K = -\K \lambda w_\lambda^\K$ in the distributional sense. By Lemma \ref{lem2} it follows that $w_\lambda^\K  = 0$. From the uniqueness of the complex Laplace transform we conclude that $u=0$ a.e.

If $v = e_-^\epsilon u'_0,$ then
\beqn
||u_0-u_0'|| \leq ||e_-^\epsilon u'_0 - u'_0|| + ||u_0 - e_-^\epsilon u_0|| + ||e_-^\epsilon u_0 - e_-^\epsilon u'_0||. \label{11*}
\eeqn

The first two terms tend to zero if $t \rightarrow 0,$ the third term equals to zero by hypothesis. Hence $u'_0 = u_0.$
\end{proof}

\begin{remark}
As we already have observed, the semigroup theory provides an elegant method for establishing existence and uniqueness results for the regularized Schr\"odinger problem. However, it is important to remark that the application of this theory was only possible since the coefficients are time-independent. In the case where the coefficients of the operator are time-dependent we would need to implement a Galerkin method (for more details see Section 7.1, \cite{Eva}).
\end{remark}

\par\medskip\par

As in Subsection \ref{subsec3}, we can extend the previous results to the setting of differential forms and can consider an arbitrary $(n+1)-$Minkowski manifold. Also here, we will need to impose additional technical  conditions concerning the positiveness of the curvatures of the manifold $M$.

In the case of differential forms we need to impose that $\mathrm{Ric} > 0$. In the case of the G\"unt\-er derivatives we need to impose that $2R^2 - \mathcal{G}R >0$. With these two additional conditions and taking into account the relations (\ref{c1.5*}) and (\ref{c1.4*}), we can establish analogous proofs and may conclude that in the case of differential forms the regularized Schr\"odinger problem is solvable when $v(\cdot, t) \in L_p(\bigwedge_k M)$ and $u_0 \in L_p(\bigwedge_k M),$ with $\ds \frac{3}{2} < p < 3$, independently of the choice of considering the Bochner or G\"unter-Laplacian.


\section{The general case}

It remains to study the behavior of our results when $\epsilon$ tends to zero. The implemented regularization procedure allowed us to compute a solution of the regularized Schr\"odin\-ger equation in a stable way and to obtain a solution similar to the solution of the Schr\"odinger problem
\beqn
\left\{
\begin{array}{ll}
(-\Delta - i \partial_t) v = 0, & \quad \mbox{ on } \Omega  \\
v(x,0)=u_0(x), & \quad \mbox{ on } \ul{\Omega}
\end{array}
\right.
\label{13*}
\eeqn
when $\epsilon$ is small.

Applying the regularization procedure described in Subsection  \ref{subsec2.2}, the family of operators $-\Delta - \K \partial_t$ converges to $-\Delta-i\partial_t$ when $\epsilon$ tends to zero. In the same subsection it was indicated that the elements of the family are hypoelliptic operators, while the Schr\"odin\-ger operator is not. This fact implies that the results presented in Section 3 cannot be adapted directly to the Schr\"odinger operator because they depend on the hypoellipticity of the operator.

However, taking into account \cite{Isa} (Section 2.4), we can say that our regularization procedure corresponds to a stabilizing functional for the Schr\"odinger operator, where $\epsilon$ is the regularization parameter. Hence we can present existence and uniqueness results for the solution of problem (\ref{13*}) (which are the correspondent for the general case of Theorems \ref{teo2} and \ref{teo3})
\begin{theorem}
The initial value problem (\ref{13*}) is solvable, with $v(\cdot, t) \in L_p(\Omega)$, whenever $u_0 \in L_p(\Omega)$ and $\ds \frac{3}{2} < p < 3.$
\end{theorem}
\begin{theorem}
Let $v=v(x,t)$ be a solution of the Schr\"odinger equation with $v(\cdot,t) \in L_p(\Omega)$ and $\ds \frac{3}{2} < p < 3.$ Assume further that $||v(\cdot,t)||_p \leq a e^{-b t}$. Then there exists a uniquely determined $u_0 \in L_p(\Omega)$ such that $v = \Gamma_t u_0.$
\end{theorem}

\textbf{Acknowledgement: } The second author wishes to express his gratitude to \textit{Funda\c c\~ao para a Ci\^encia e a Tecnologia} for the support of his work via the grant \texttt{SFRH/BPD/73537/2010}.

The third author wishes to express his gratitude to \textit{Funda\c c\~ao para a Ci\^encia e a Tecnologia} for the support of his work via the grant \texttt{SFRH/BPD/65043/2009}.


\end{document}